\documentclass[%
reprint,
amsmath,amssymb,
aps,
]{revtex4-1}

\usepackage{graphicx}% Include figure files
\usepackage{dcolumn}% Align table columns on decimal point
\usepackage{bm}% bold math

\usepackage{amsmath,amssymb,amsfonts}
\usepackage{graphicx}
\usepackage{color}
\usepackage{appendix}
\usepackage{amsthm}

\newtheorem{theorem}{Theorem}
\newtheorem{corollary}{Corollary}
\newtheorem{lemma}{Lemma}
\newtheorem{definition}{Definition}

\newcommand{\be}{\begin{eqnarray}}
\newcommand{\ee}{\end{eqnarray}}
\newcommand{\nn}{\nonumber}
\newcommand{\bn}{\begin{enumerate}}
	\newcommand{\en}{\end{enumerate}}
\newcommand{\bl}{\begin{align}}
\newcommand{\el}{\end{align}}

\def\a{\alpha}

\def\d{\delta}
\def\eps{\epsilon}

\def\la{\lambda}

\def\r{\rho}

\def\p{\psi}

\def\La{\Lambda}
\def\P{\Psi}
\def\iff{\Longleftrightarrow}

\def\<{\langle}
\def\>{\rangle}

\def\jmath{{j}}

\begin{document}
\title{Coherence number as a discrete quantum resource}	 
\author{Seungbeom Chin}
\email{sbthesy@skku.edu}
\affiliation{College of Information and Communication Engineering, Sungkyunkwan University, Suwon 16419, Korea}

%\date{\today}

\begin{abstract}
We introduce a new discrete coherence monotone named the \emph{coherence number}, which is a generalization of the coherence rank to mixed states.
After defining the coherence number in a similar manner to the Schmidt number in entanglement theory, we present a necessary and sufficient condition of the coherence number for a coherent state to be converted to an entangled state of nonzero $k$-concurrence (a member of the generalized concurrence family with $2\le k \le d$).
It also turns out that the coherence number is a useful measure to understand the process of Grover search algorithm of $N$ items. We show that the coherence number remains $N$ and falls abruptly when the success probability of the searching process becomes maximal. 
This phenomenon motivates us to analyze the depletion pattern of $C_c^{(N)}$ (the last member of the generalized coherence concurrence, nonzero when the coherence number is $N$), which turns out to be an optimal resource for the process since it is completely consumed to finish the searching task.
\begin{description}
	\item[PACS numbers] 03.67.−a, 03.65.Ud, 03.67.Ac
\end{description}
\end{abstract}

	\maketitle
	
	%\tableofcontents

\section{Introduction}
\indent

Coherence is a fundamental property of quantum mechanics that generates several intrinsic features distinguished from classical ones. It is also useful as a physical resource for some quantum information processes. To perform the quantitative analysis of these tasks, we need rigorous definitions and formulations of the coherence resource theory. The first comprehensive formulation was presented in \cite{baum}, where the authors provided strict criterions for a quantity to  be a measure for the amount of coherence. It was a milestone from which productive studies on coherence resource theory thrived in varied areas, e.g., discovering measures and monotones of coherence \cite{yuan, winter, napoli, piani, tan, chitambar}, comparing coherence with other quantum correlations \cite{strel, adesso, roga, ma, mmarv, marv}, dynamics of coherence \cite{bromley, mani, puchala, singh, pires, mondal}, and application to quantum thermodynamics \cite{lost, cwik, naras} (a recent review on the developing landscape of the coherence resource theory is given in \cite{SP}).

Among them, one of the principal tasks is to delve into the connection between coherence and entanglement theory. It was shown that nonzero coherence is a necessary and sufficient condition for a state to be used to generate  entanglement \cite{strel}. This result was generalized to a wider category  in \cite{killoran}, which analyzed an extended form of \emph{nonclassicality}. The authors presented a framework for the conversion of nonclassicality (including coherence) into entanglement. The entanglement convertibility theorems have two distinctive scenarios, discrete (in which the classical states are in a finite linearly independent set) and continuous (in which the states are named symmetric coherent states connected with the SU$(K)$ representation). Especially in the discussion of the discrete case, an analogous concept to the Schmidt rank in entanglement is introduced, which is the $coherence$ $rank$ of pure states.

In this paper, we generalize the concept of coherence rank for pure states to one that is suitable for mixed states, which is the \emph{coherence number} $r_C(\r)$.  It is a discrete coherence monotone and defined in a similar manner to the construction of the Schmidt number from the Schmidt rank \cite{terhal}. So $r_C(\r)$ is the smallest possible maximal coherence rank in any decomposition of a mixed state $\r$. We expect that it will be a simple but useful tool for measuring the coherence variations in many quantum processes.
%The coherence number for mixed states is defined as 
%\begin{align}
%r_C(\r) \equiv \min_{\{(p_a,|\p_a\>)\}}\max_a\Big[r_C(|\p_a\>)\Big],
%\end{align}
%i.e., 

As the first application, we investigate with coherence number the generalized concurrence monotone in the perspective of the entanglement convertibility theorem.  The generalized concurrence monotone is a family of entanglement monotones for $(d\times d)$-systems \cite{gour}, which includes the entanglement concurrence for $(2\times 2)$-systems \cite{hill, wooters} and its higher-dimensional generalization \cite{rungta}. It is worth investigating as a candidate for the quantity to witness the entanglement dimensionality concretely \cite{sentis}. The members of the concurrence family (called the $k$-concurrence and denoted as $C_k$) have a strict quantitative order, and especially the $G$-concurrence (the $k$-concurrence for $k=d$) has convenient mathematical features such as multiplicativity. It will be shown in our discussion that \emph{a mixed state $\r$ can be converted to a state of nonzero $k$-concurrence if and only if $r_C(\r) \ge k$}.

%We use a coherence monotone named the $coherence$ $concurrence$  \cite{Qi} for the quantitative comparison of the $k$-concurrence with the coherence.
%The coherence concurrence is a convex roof coherence monotone based on the generalized Gell-Mann matrices. It is convenient to compare with $C_k$, for most inequalities between $C_k$ and $C_c$ for pure states are also valid for mixed states from the fact that both are convex roof quantities. It is proved in \cite{Qi} that the coherence in a quantum system $S$ can be converted to $C_2$ between $S$ and an ancilla system by some incoherent operation. We search the condition for a state to be converted to a nonzero $k$-concurrence state.
 
 Next, we discuss the role of coherence number in the Grover search algorithm \cite{grover1997}. Coherence is assumed to be a fundamental quantum resource which has the most obvious correlation with the success probability $P$ of Grover search process \cite{shi}. We show that the coherence number is a convenient measure for detecting the fulfillment of the searching task with $N$ items. The coherence number for the state remains $N$ and sharply drops off when $P=1$. This pattern motivated us to analyze the behavior of $C_c^{(N)}$, the last member of the generalized coherence concurrence family defined in \cite{Chin2}, during the process. $C_c^{(N)}$ is a normalized quantity and nonzero if and only if the coherence number is $N$. The advantage of $C_c^{(N)}$ as a resource for Grover algorithm over other coherence monotones is that it monotonically decreases as $P$ increases and completely disappears when $P =1$. So we can state that \emph{$C_c^{(N)}$ is an optimal coherence monotone for Grover algorithm}.

This paper is organized as follows. 
In Section \ref{revisit},  we review the Schmidt number and the generalized concurrence in entanglement resource theory.  We derive an expression for $C_k$ which we can use to obtain some bounds of $C_k$.
In Section \ref{cohno}, we introduce the coherence number and show that it is a discrete coherence monotone.
In Section \ref{conversion}, we use the concept of coherence number to study the entanglement convertibility theorem of the $k$-concurrence monotone.
In Section \ref{groversearch}, we show that the coherence number and $C_c^{(N)}$ are good resources for the Grover searching process which clearly reveal critical moments of the process.
In Section \ref{fin}, we summarize our results and discuss further problems.

\section{SCHMIDT NUMBER AND GENERALIZED CONCURRENCE REVISITED}\label{revisit}

\indent

In this section, we briefly review the concepts of the Schmidt number and the generalized entanglement concurrence. Then we derive an expression for $C_k$ that will be used for the quantitative analysis in Section \ref{conversion}.

The Schmidt coefficients are key elements to both entanglement monotones. 
Considering a quantum system with two subsystems $A$ and $B$ (dim $\mathcal{H}_A=d$, dim $\mathcal{H}_B=d'$ and $d \le d'$), a pure state $|\p\>$ $\in$ $ \mathcal{H}_A \otimes \mathcal{H}_B$ is always possible to be written as  
\begin{align}
|\p\> = \sum_i^{r(|\p\>)}\sqrt{\la_i}|i\tilde{i}\>_{AB},
\end{align} 
with orthonormal bases $\{ |i\>_A\}_{i=1}^{d}$ and $\{|\tilde{i}\>_B\}_{i=1}^{d}$. The real positive numbers $\la_i$ are the Schmidt coefficients of $|\p\>$. And the number of nonzero Schmidt coefficients, $r(|\p\>)$, is the Schmidt rank of $|\p\>$.

 The $Schmidt$ $number$ $r(\r)$ is an extension of Schmidt rank to mixed states \cite{terhal}. It is defined as
 \begin{align}
 r(\r) = \min_{\{p_i,|\p_i\>\}}\max_{i}\Big[r(|\p_i\>)\Big],
 \end{align}    
where $\{p_i, |\p_i\> \}$ is the set of all possible pure-state decompositions of $\r$.  So we choose one pure state decomposition which has the smallest maximal Schmidt rank.

The \emph{generalized concurrence monotone} is a family of entanglement monotones for $(d\times d)$-systems \cite{gour}, which is a generalization of the entanglement concurrence for $(2\times 2)$-systems \cite{hill, wooters}. It consists of $k$-concurrences with $2\le k \le d$. Considering a $(d\times d)$-dimensional bipartite pure state $|\p\> = \sum_i\sqrt{\la_i}|i\tilde{i}\>_{AB}$, the $k$-concurrence of $|\p\>$ is defined as  
\begin{align}
\label{ckp}
C_k(|\p\>) &\equiv \Big[\frac{S_k(\la)}{S_k(1/d,1/d, \cdots , 1/d)}\Big]^\frac{1}{k}, \\
S_k(\la) &\equiv \sum_{i_1 < i_2 < \cdots< i_k}\la_{i_1}\la_{i_2}\cdots \la_{i_k}, 
\end{align}
where $\la=(\la_1,\la_2,\cdots \la_d)$.
$S_k(1/d, \cdots, 1/d) =\frac{1}{d^k}\binom{d}{k}$ is in the denominator so that $C_k(|\p\>)$ is normalized as $0 \le C_k(|\p\>) \le 1$. $C_k(|\p\>)$ equals $1$ only when $|\p\>$ is maximally entangled.

The $k$-concurrence for a mixed state $\r$ is defined by convex roof extensions, i.e., 
\begin{align}
& C_k(\r) \equiv \min_{\{p_i,|\p_i\>\}} \sum_i p_iC_k(|\p_i\>) \nn \\
& \Big(\r =\sum_ip_i|\p_i\>\<\p_i|, \quad \sum_ip_i=1, \quad p_i\ge 0\Big).
\end{align}
This form of concurrence family contains the entanglement monotonones that exist only in $(d\times d)$-dimensional systems with $d>2$.
 
The $k$-concurrence is zero when $k$ is larger than the Schmidt number of the state, which means that the generalized concurrence is a \emph{Schmidt number specific} monotone family. 

The $G$-concurrence  $G_d$ is the last member of the $k$-concurrence family, i.e.,  $G_d=C_{d}$ ($G$ stands for the geometric mean of the Schmidt coefficients). It has some convenient properties. For example, with two bipartite entangled states $|\p_1\>$ and $|\p_2\>$ of dimension $d_1\times d_1$ and $d_2\times d_2$, we have
\begin{align}
G_{d_1d_2}(|\p_1\>\otimes|\p_2\>)=G_{d_1}(|\p_1\>)G_{d_2}(|\p_2\>),
\end{align}
which follows directly from the property of the determinant. More important is that $G_d$ provides a lower bound for the $k$-concurrence family. For mixed bipartite states we have
\begin{align}
\label{GC}
G_d(\r) \le C_k(\r) \quad \forall k=1,2,\cdots ,d.
\end{align} 
We can derive this inequality using the arithmetic-geometric mean inequality.  The $G$-concurrence monotone measures to which extent pure states with maximal Schmidt rank is contained in a mixed state, and is useful to analyze some entanglement system, e.g., remote entanglement distribution (RED) protocols. For more details, see \cite{gour, sentis} and 5.2.2 of \cite{eltschka}.

Now we rewrite \eqref{ckp} in terms of $|\p\>$ that is not Schmidt-decomposed,
\begin{align}
|\p\> = \sum_{ij}\p_{ij}|ij\>_{AB},
\end{align} 
and derive the explicit formula for $C_k(\p_{ij})$. By definition
the pure state $k$-concurrence is given by
\begin{align}
C_k(|\p\>) =d\Bigg[\frac{Tr\Big(K_k(\P^\dagger \P)\Big)}{\binom{d}{k}}\Bigg]^\frac{1}{k},
\end{align}
where $(\P)_{ij}= \p_{ij}$ and $K_k(\P^\dagger \P)$ is the $k$th compound matrix of $\P^\dagger \P$ \cite{gour}. Using Cauchy-Binet formula
\begin{align}
K_k(\P^\dagger \P)=K_k(\P^\dagger)K_k(\P)
\end{align}
and $K_k(\P^\dagger)=K_k(\P)^\dagger$, we can obtain the explicit expression of $C_k(|\p\>)$ in terms of $\p_{ij}$ as follows:
\begin{align}
&C_k(|\p\>) \nn \\
&=d\Bigg[\frac{1}{\binom{d}{k}}\sum_{\substack{i_1< \cdots <i_k \\j_1 < \cdots <j_k}} \Bigg|\sum_{a_1,\cdots, a_k=1}^{k} \eps_{a_1\cdots a_k} \p_{i_1j_{a_1}} \cdots\p_{i_kj_{a_k}} \Bigg|^2 \Bigg]^{\frac{1}{k}}
\end{align}
This formula will be used in Section \ref{conversion} to obtain some bounds for $C_k$.\\
\\
\emph{(Example)}\\
$k=2$:\\
\begin{align}
C_2(|\p\>)&=\Big( \frac{2d}{(d-1)}\sum_{\substack{i_1<i_2 \\ j_1<j_2 }} \Bigg| \sum_{a_1,a_2=1}^{2}\eps_{a_1a_2} \p_{i_1j_{a_1}}\p_{i_2j_{a_2}}\Bigg|^2 \Big)^{\frac{1}{2}} \nn \\
 & = \Big( \frac{2d}{(d-1)}\sum_{\substack{i_1<i_2 \\ j_1<j_2 }} \Big| \p_{i_1j_1}\p_{i_2j_2} - \p_{i_1j_2}\p_{i_2j_1} \Big|^2 \Big)^{\frac{1}{2}},
\end{align}
which equals Eq. (22) of \cite{akh}.
\\
\\
$k=3$:\\ 
\begin{align}
\label{c3}
&C_3(|\p\>) \nn \\
& = \Bigg[\frac{3!d^2}{(d-1)(d-2)} \nn \\&\quad \times \sum_{\substack{i_1<i_2<i_3 \\j_1<j_2<j_3}} \Bigg| \sum_{a_1, a_2,a_3=1}^{3}\eps_{a_1a_2a_3}\p_{i_1j_{a_1}}\p_{i_2j_{a_2}}\p_{i_3j_{a_3}}\Bigg|^2 \Bigg]^{\frac{1}{3}}.
\end{align}
The explicit expension and application of \eqref{c3} is given in Appendix \ref{cc3}.
\\
\\
$k=d$: 
\begin{align}
C_d(|\p\>)=G_d(|\p\>) =d\Bigg|\sum_{a_1 a_2\cdots a_d=1}^{d}\eps_{a_1,a_2,\cdots ,a_d}\p_{1a_1}\cdots\p_{da_d}\Bigg|^\frac{2}{d} .
\end{align}
In this case we can easily see that the following relation holds as expected: 
\begin{align}
G_d(|\p\>) =d(Det(\P^\dagger \P))^\frac{1}{d} = d(S_d(\la))^\frac{1}{d}.
\end{align}

\section{Definition of Coherence Number}\label{cohno}
\indent

The coherence resource theory has developed along the landscape of the entanglement resource theory, exhibiting strong correspondences in many aspects. Streltsov $et$ $al.$ \cite{strel} proved that any coherent state can be converted to a bipartite entangled state by adding an ancilla and taking incoherent operations. The similar process for quantum discord is presented in \cite{ma}. 

The conversion of coherence to entanglement is generalized to a wider category  by \cite{killoran}, who analyzed the nonclassicality including coherence.
During the discussion they introduced an analogous concept to the Schmidt rank in entanglement, the $coherence$ $rank$ of a pure state:
\begin{align}
r_C(|\p\>) \equiv \min \Bigg\{  r \Bigg| |\p\> =\sum_{j=1}^{r\le d}\p_j|c_j\> \Bigg\}, 
\end{align}
where $|c_j\>$ are in the set of computational basis and each classical, and $\forall j: \p_j \neq 0$. So $1 \le r_C \le d$ and all nonclassical pure states should have $r_C \ge  2$. It is proved that there exists a unitary incoherent operation $\La$ on a pure state $|\p\>$ such that the Schmidt rank of $\La|\p\>$ is equal to the coherence rank of $|\p\>$ \cite{killoran}, and  $r_C(|\p\>)$  is non-increasing under incoherent operations \cite{winter, vicente}.

It is not hard to conceive generalized concepts of coherence rank to mixed states. One possible way is to build a similar quantity to the Schmidt number introduced in Section \ref{revisit} as follows:

\begin{definition}\label{cohnum}
	The coherence number $r_C(\r)$ for a mixed state $\r$ is defined as 
	\begin{align}
	r_C(\r) \equiv \min_{\{(p_a,|\p_a\>)\}}\max_a\Big[ r_C(|\p_a\>)\Big].
	\end{align}
\end{definition} 

So $r_C(\r)$ is the smallest possible maximal coherence rank in any decomposition of the mixed state $\r$, and for pure states the coherence number equals the coherence rank. It is obvious that there exists a unitary incoherent operation $\La$ on a mixed state $\r$ such that the Schmidt number of $\La[\r]$ is equal to $r_C(\r)$.

If we denote the set of states on $\mathcal{H}_d$ that have coherence number not bigger than $k$ as $R_k$, i.e.,
\begin{align}
\forall \r \in R_k, \quad r_C(\r) \le k.
\end{align}
 then $R_{k-1} \subset R_k $ and $R_k$ is a convex compact subset of the entire set of states $R_d$,  just as the set of quantum states on $\mathcal{H}_d\otimes \mathcal{H}_d$ that have Schmidt number not bigger than $k$ is a convex compact subset of the entire set of states \cite{terhal}.

\begin{theorem}
	The coherence number  $r_C(\r)$ (or $\log_2[r_C(\r)]$ for the quantity to be zero when incoherent) is a coherence  monotone, which satisfies the condition (C1), (C2) and (C3) listed in Appendix \ref{off}.
\end{theorem}
\begin{proof}
	(C1) It is clear from Definition \ref{cohnum} that $r_C(\r)$ is not negative, and 1 if and only if $\r$ is incoherent.\\
	(C2) Let's consider that $r_C(\r)$ for a mixed state $\r$ is $l$. If $r_C(\La[\r])$ is bigger than $l$, there exists a decomposing pure state $|\phi\>$ of $\La[\r]$  such that $r_C(|\phi\>) >l$. This means that $\r$ can be decomposed as to include a pure state which has the coherence rank bigger than $l$, so $r_C(\r) >l$. So $r_C(\La[\r])$  cannnot be bigger than $l$. \\
	(C3) $\forall n$: $r_C(\r) \ge r_C(K_n\r K_n^\dagger)$ with Definition \ref{cohnum} shows that the strong monononicity holds for $r_C(\r)$.
\end{proof}

The conditions for coherence monotones to satisfy along the incoherent operations are listed in Appendix \ref{off}.

\section{Measuring the convertibility of coherence into $k$-concurrence  with $r_C$}\label{conversion}

We expect that the coherence number will be a simple but useful criterion for recognizing the non-classicality of general quantum states as the Schmidt number does in the entanglement resource theory.
In this section, we compare the coherence concurrence of a mixed state $\r^s$ in an initial system $S$ with the $k$-concurrence entanglement generated from $\r^s$ by attaching an ancilla system $A$ (of the same dimension with the system $S$) and taking an incoherent operation $\La^{SA}$. It will be shown that a state $\r$ can be converted to an entangled state of nonzero $k$-concurrence if and only if $r_C(\r) \ge k$.

%Some identities for the expression of the pure state coherence concurrence is summarized in Appendix \ref{id}, which is useful for the calculations in this section.

\subsection*{An coherence upper bound of $k$-concurrence monotones }
\indent

Before approaching the main task, we first present an upper bound of the generalized entanglement monotone family created from $\r^s$ by an incoherent operation, which is given by  the \emph{coherence concurrence}, recently proposed in \cite{Qi}. We denote it $C_c$ \footnote{We would like to emphasize that $C_c$ is quantitatively different from the \emph{generalized coherence concurrence } $C_c^{(k)}$ with $2\le k \le d$ introduced in \cite{Chin2}.}.

For a pure $|\p\>=\sum_{i}\p_i|i\>$ ($\{|i\>\}_{i=1}^d$ is the computational basis set and all incoherent density operators are of the form $\r=\sum_{i=1}^{d}p_i|i\>\<i|$),  the coherence concurrence is defined as
\begin{align}
C_c(|\p\>) =\sum_{j<k}|\<\p|\La^{j,k}|\p\>| =2\sum_{j<k}|\p_j\p_k|,
\end{align}
where $\La^{j,k} \equiv |j\>\<k|+|k\> \<j|$ $(1\le j<k\le d)$. We can consider $\La^{j,k}$ as the symmetric generators of SU($d$) group (GGM, the generalized Gell-Mann matrices).  
For a mixed state $\r$, the coherence concurrence $C_c(\r)$ is defined with convex roof construction. In general $C_c$ is not smaller than $C_{l_1}$ ($l_1$-norm coherence monotone), but there exists a necessary and sufficient condition for the two quantities to be equal to each other \cite{Chin2}.

Then the $2$-concurrence entanglement monotone created from $\r^s$ is bounded by $C_c(\r^s)$:\\
(\emph{Theorem 2} in \cite{Qi}) The amount of $2$-concurrence entanglement monotone created from $\r^s$ (a state in the system $S$ of the dimension $d$) by adding an incoherent state $|1\>\<1|^A$ in an ancilla system $A$ and taking an incoherent operation $\La^{SA}$, is bounded above by the coherence concurrence of $\r^s$ as follows \footnote{Note that the factor $\sqrt{\frac{d}{2(d-1)}}$ in front of $C_c$ is by our different normalization from that of \cite{Qi} }:
\begin{align}
\label{C_2}
C_{2}(\La^{SA} [ \r^s\otimes |1\> \<1|^A]) \quad \le \quad \sqrt{\frac{d}{2(d-1)}}C_c(\r^s).
\end{align}

A similar inequality holds for $k=3$ $(d\ge 3)$ case, e.g.,
\begin{align}
\label{Cc3}
C_3(\La^{SA} [ \r^s\otimes |1\> \<1|^A]) \le \Big(\frac{3d^2}{4(d-1)(d-2)}\Big)^\frac{1}{3}C_c(|\p\>),
\end{align}
using the formula \eqref{c3}. The detailed proof is in Appendix \ref{cc3}.

But there is a simpler way to obtain a complete inequality relation of $k$-concurrence monotones that has the upper bound in terms of $C_c(\r)$ from \eqref{C_2} and the following inequality,\\
\begin{align}\label{cf}
C_2(\r) \ge C_3(\r) \ge \cdots \ge C_{d-1}(\r) \ge C_d(\r) \equiv G_d(\r)
\end{align}
for any mixed bipartite state $\r$, which is a direct result of Maclaurin's inequality and convex roof extention.
\begin{theorem} The members of the $k$-concurrence monotone family created from any mixed state $\r^s$ via an incoherent operation $\La^{SA}$ is bounded above by $C_c(\r^s)$ and ordered as follows:
\begin{align}
\label{gd}
 G_{d}(\La^{SA} [ \r^s\otimes |1\> \<1|^A]) &\le C_{d-1}(\La^{SA} [ \r^s\otimes |1\> \<1|^A]) \nn \\ 
& \le \cdots \le C_{2}(\La^{SA} [ \r^s\otimes |1\> \<1|^A]) \nn \\
&\le \sqrt{\frac{d}{2(d-1)}}C_c(\r^s). 
\end{align}
\end{theorem}
\begin{proof}
 This is a straightforward result of \eqref{C_2}, \eqref{gd}, and the inequality $\sqrt{\frac{d}{2(d-1)}} < \Big(\frac{3d^2}{4(d-1)(d-2)}\Big)^\frac{1}{3}$ for $d\ge 3$. 	
\end{proof}
\begin{corollary}\label{cor1}
If there exists an incoherent operation that converts a state $\r^s$ to a state of nonzero $k$-concurrence for any $k$, $C_c(\r^s)$ is nonzero.
\end{corollary}

\subsection*{The conversion of coherence into $k$-concurrence }
The generalized concurrence is a family of hierarchical entanglement monotones closely related to the Schmidt number of the state, so we can guess the convertibility for each $k$-concurrence ($2 \le k \le d$) will be discernable with some hierarchical coherence monotone. And we presume that the coherence number is such a quantity.

%For the case of $2$-concurrence, we can state that

%\emph{There exists an incoherent operation that converts a state $\r^s$ to a state of nonzero $2$-concurrence if and only if $C_c(\r^s)$ is nonzero}

%following \cite{Qi}, by using a unitary incoherent operation on the bipartite system \cite{strel, stein}
%\begin{align}
%\label{u}
%U\equiv \sum_{i=1}^{d}\sum_{j=i}^{d}|i\>\<i|^S\otimes |i\oplus (j-1)\>\<j|^A,
%\end{align}
%where $\oplus$ means an addition modulo $d$. For bipartite qubits this is the CNOT gate. In general
%\begin{align}
%\La^{SA}_u\Big[ \r^s\otimes |1\>\<1|^A\Big] & \equiv U\Big[ \r^s\otimes |1\>\<1|^A\Big]U^\dagger \nn \\
%& = \sum_{i,j}\r_{ij}|i\>\<j|^S\otimes |i\>\<j|^A
%\end{align}  
%and $|\p\>^S=\sum_{i=1}^{d}\p_i|i\>$ goes to $|\p\>^{SA}=\sum_{i=1}^{d}\p_i|ii\>$ under $\La^{SA}_u$.
%Then we have a lower bound for  $C_2(\La^{SA}_u [ \r^s\otimes |1\> \<1|^A])$ \cite{Qi} as follows:
%\begin{align}
%\label{uni2}
%\frac{1}{(d-1)}C_c(\r^s) \quad \le \quad C_2(\La^{SA}_u [ \r^s\otimes |1\> \<1|^A]).
%\end{align}
%Combining \eqref{uni2} and Corollary \ref{cor1}, we obtain the above statement.

We can obtain the convertibility relation between the coherence number and $k$-concurrence entanglement of a state by imposing a constraint on the coherence number of the state through the following lemma:
\begin{lemma}\label{cnsn}
The Schmidt rank generated from a pure state in the system $S$ through any Kraus operator of incoherent operations by appending an incoherent state
$|1\>^A$ in an ancilla system $A$ is not bigger than the coherence rank of the initial pure state, i.e.,
\begin{align}
r(K_n[|\p\>^S\otimes |1\>^A]) \le r_C(|\p\>^S).
\end{align}
\end{lemma}
\begin{proof}
Let's say that a pure state $|\p\>^S$ in $S$ has a coherence rank $l$. Then with the Kraus operator set $\{K_n\}$ of any incoherent operation $\La^{SA}$ acting on $S$ and $A$, we have
\begin{align}
l= r_C(|\p\>^S) =r_C(|\p\>^S\otimes|1\>^A) \ge r_C(K_n[|\p\>^S\otimes |1\>^A]) 
\end{align}
for all $n$.
So $K_n[|\p\>^S\otimes |1\>^A]$ can be rewritten as
\begin{align}
K_n[|\p\>^S\otimes |1\>^A]=  \sum_{i=1}^{q\le  l} |i\>\otimes(\sum_j (\p_n)^{ij}|j\>) \equiv \sum_i|i\>\otimes|\tilde{i}\>
\end{align}
and the Schmidt rank of $K_n[|\p\>^S\otimes |1\>^A]] $ is not bigger than $l$.
\end{proof}
Now we are ready to present the convertibility theorem between coherence and the $k$-concurrence of general states.
\begin{theorem}
A mixed state $\r_s$ can be converted to a state of nonzero $k$-concurrence via an incoherent operation by appending an incoherent state
$|1\>\<1|^A$ in an ancilla system $A$ if and only if $r_C(\r^s) \ge k$, i.e.,
\begin{align}
^\exists \La^{SA}: \quad  C_k(\La^{SA}[\r^s\otimes |1\>\<1|^A]) \neq 0 \quad \iff\quad  r_C(\r^s) \ge k.
\end{align}
\end{theorem} 
\begin{proof}
$\Longrightarrow$: If $r_C(\r^s) < k$, there exists a decomposition of $\r^s$ as $\r^s=\sum_ap_a|\p_a\>\<\p_a|$ such that the maximal coherence rank of $\{|\p_a\>\}$ is smaller than $k$. Then by Lemma \ref{cnsn}, the Schmidt number of $\La^{SA}[|\p_a\>\otimes |1\>^A]$ is smaller than $k$. So we have
\begin{align}
C_k(\La^{SA}[|\p_a\>^S\otimes |1\>^A])=0,\quad \forall a.
\end{align}
Hence, 
\begin{align}
0 & \le C_k(\La^{SA}[\r^s\otimes |1\>\<1|^A]) \nn \\
& \le \sum_ap_aC_k(\La^{SA}[|\p_a\>^S\otimes |1\>^A])=0.
\end{align}
gives $C_k(\La^{SA}[\r^s\otimes |1\>\<1|^A])=0$\\
$\Longleftarrow$:
If $r_C(\r^s) \ge k$, then there exists an incoherent operation $\La^{SA}$  under which the coherence number of initial states are equal to the Schmidt number of final states (which is clear from Theorem 1 of \cite{killoran}). So there exists an incoherent operation $\La^{SA}$ such that $C_k(\La^{SA}[\r^s\otimes|1\>\<1|]) \neq 0$.
\end{proof}
An unitary operation under which the coherence number and the Schmidt number are equal is given by
\begin{align}
\label{u}
U\equiv \sum_{i=1}^{d}\sum_{j=i}^{d}|i\>\<i|^S\otimes |i\oplus (j-1)\>\<j|^A,
\end{align}
where $\oplus$ means an addition modulo $d$. Then we have with $|\p\>^S = \sum_i \p^i|i\>^S$
\begin{align}
U[ |\p\>^S\otimes |1\>^A] = \sum_{i}\p^i|ii\>^{SA}.
\end{align}

%\begin{corollary}
%The possible largest set of nonzero $C_k(\La^{SA}[\r^s\otimes|1\>\<1|^A])$ for a state $\r^s$ in $R_l$ is $\{C_2(\r^s), C_3(\r^s), \cdots, C_l(\r^s)  \}$.
%\end{corollary}
Defining an unitary incoherence operation as 
\begin{align}
\La^{SA}_u\Big[ \r^s\otimes |1\>\<1|^A\Big] & \equiv U\Big[ \r^s\otimes |1\>\<1|^A\Big]U^\dagger, 
\end{align}
we can obtain the bounds of $G$-concurrence ($C_d$) with coherence as follows:
\begin{theorem}\label{gcon}
	 When $r_C(\r^s)=d$ for a mixed state $\r^s$ and the unitary incoherent operation is given as $U$ of \eqref{u}, $G_d(\La^{SA}_u [ \r^s\otimes |1\> \<1|^A])$ has the upper and lower bound as follows:
\begin{align}
\label{Gdineq}
\frac{C_c(\r^s)}{S(\eps)(d-1)} \le G_d(\La^{SA}_u [ \r^s\otimes |1\> \<1|^A]) \le \frac{C_c(\r^s)}{(d-1)},
\end{align}
where
\begin{align}
&\r=\sum_ap_a|\p_a\>\<\p_a|, \quad |\p_a\> =\sum_i\p_a^i|i\>, \nn \\ & |\p_a^i| \ge \eps \quad \textrm{for any possible decomposition of }\r \nn
\end{align} and
\begin{align}
S(\eps)\equiv \frac{1}{e} \Bigg(\frac{(\eps^2 -1)\eps^{\frac{2\eps^2}{\eps^2 -1}}}{2\eps^2 \log\eps}\Bigg) \le \frac{1}{\eps^2}. \nn
\end{align}	
\end{theorem}	
\begin{proof}
For a pure state $|\p\>^S$, $G$-concurrence and coherence concurrence in terms of $r_i$ ($\equiv |\p_i|$) are given by
\begin{align}
G_d(|\p\>^{SA})=d\Big( \prod_{i=1}^{d}  r_{i}^2 \Big)^{\frac{1}{d}}, \qquad 
C_c(|\p\>)=2\sum_{i<j}r_{i}r_{j}.
\end{align}
Then we have
\begin{align}
&\frac{1}{d(d-1)}C_c(|\p\>^S) - \frac{1}{d}G_d(|\p\>^{SA}) \nn \\
&=\frac{2}{d(d-1)}\sum_{i<j}r_ir_j -\Big(\prod_{i=1}^{d}r_i \Big)^{\frac{2}{d}} \nn \\
&=\frac{2}{d(d-1)}\sum_{i<j}r_ir_j -\Big(\prod_{i<j}r_ir_j \Big)^{\frac{2}{d(d-1)}} \nn \\
&\ge 0 ,
\end{align}
where the last inequality holds by the arithmetic-geometric mean inequality, and 
\begin{align}
& \frac{S(\eps)}{d}G_d(|\p\>^{SA})- \frac{1}{d(d-1)}C_c(|\p\>^S) \nn \\ &=S(\eps)\Big(\prod_{i<j}r_ir_j \Big)^{\frac{2}{d(d-1)}}  -\frac{2}{d(d-1)}\sum_{i<j}r_ir_j \nn \\
&\ge 0,
\end{align}
since $S(\eps)$ reverses  the arithmetic-geometric mean inequality \cite{Specht}. By convex roof extension we obtain \eqref{Gdineq}. 
\end{proof}
%It is interesting that the upper bound of  $G_d(\La^{SA}_u [ \r^s\otimes |1\> \<1|^A])$ is equal to the lower bound of  $C_2(\La^{SA}_u [ \r^s\otimes |1\> \<1|^A])$ which is derived in \cite{Qi}.

We can compare Theorem \ref{gcon} with the results in \cite{sentis}, where the lower bound of $G_d(\r)$ is given using nonlinear witness techniques. The inequality \eqref{Gdineq} also provides a lower bound for $G$-concurrence of a state, but the direction is different. The lower bound given in \cite{sentis} is expressed with density matrix elements, so determines whether a bipartite state has nonzero $G$-concurrence. For our case, we create the entangled state with nonzero $G$-concurrence with a coherent state $\r$ with $r_C(\r)=d$.

\section{Coherence number in the Grover search algorithm}\label{groversearch}
 In this section, we show that the coherence number is a convenient measure for detecting the moment that the  Grover search process \cite{grover1997} becomes completely successful, which provides the idea that there exist optimal coherence monotones in the generalized coherence concurrence family \cite{Chin2} which are completely exploited during the task.

\subsection*{Grover search algorithm and coherence}

Grover search algorithm is the most fundamental algorithm in quantum computation. It theoretically says that quantum operations with properly adjusted phases can speed up the searching process, i.e., finding $m$ targets among a large database $N$. It is conjectured that quantum correlations such as entanglement are the resources for the speedup, but the attempts to find some concrete relation between the success probability of Grover search process and various measures of entanglement or discord has been unsuccessful \cite{braun, cui2010}. 

But considering the recent viewpoint that quantum coherence is a more fundamental resource than entanglement or discord, it is worth attempting to investigate the quantitative relation between coherence and Grover search algorithm. Indeed,  coherence depletion phenomena in the Grover quantum search algorithm is analyzed by \cite{shi}, in which the authors showed that the relative entropy of coherence and $l_1$-norm coherence monotone decrease monotonically while the success probability of the searching process increases.

 Here we approach the problem with two coherence monotones. One is the coherence number and the other is the last member of the generalized coherence concurrence $C_c^{(N)}$ introduced in \cite{Chin2}, since they expose the critical moments of the searching process more vividly than the monotones analyzed in \cite{shi}.
 
First, we briefly review the Grover search algorithm \cite{grover1997}. Consider a system with $n$-qubits. Then the system has a database of dimension $N = 2^n$. We initialize the state of the qubits as $|\p_0\> =\frac{1}{\sqrt{N}}\sum_{K=1}^{N} |K\>$, which is achieved by taking local Hadamard gates $H^{\otimes n}$ ($H=\frac{1}{\sqrt{2}}(|0\>\<0| +|0\>\<1| +|1\>\<0| -|1\>\<1|)$ ) on the ground state $|0,0,\cdots, 0\>$. Then we repeatedly take an operation $G=(2|\p\>\<\p|-\mathbb{I})O$, where $O$ is called the oracle. When the state is among the targets, $O$ rotates the phase by $\pi$. And when the state is not, $O$ leaves the system unchanged. We can easily see that $G$ rotates the state by an angle $A=\cos^{-1}{\frac{N-2m}{N}}=2\tan^{-1}{\sqrt{\frac{m}{N-m}}}$.

Let's say that there are $m$ target states among the $N=2^n$ states, Then we reexpress the initial state as 
\begin{align}
|\p_0\> =\sqrt{\frac{m}{N}}|X\> +\sqrt{\frac{N-m}{N} } |X^{\perp}\>,
\end{align} 
where $|X\>$ (for targets) and $|X^{\perp}\>$ (for those which are not) are defined as
\begin{align}
|X\> =\frac{1}{\sqrt{m}}\sum_{i=1}^{m}|i\>, \qquad |X^{\perp}\> =\frac{1}{\sqrt{N-m}}\sum_{I=m+1}^{N} |I\>, 
\end{align}
without loss of generality.
After  taking $G^r$ on $|\p_0\>$, we have 
\begin{align}
\label{pr}
|\p_r\>=\sin\a_r|X\> + \cos\a_r|X^{\perp}\>,
\end{align}
where 
$\a_r=(r+\frac{1}{2})A$.
%\begin{align}
%\r (r) =& \frac{\sin^2\a}{m}\sum_{i,j}^{m}|i \>\< j| +\frac{\cos^2\a_r}{(N-m)} \sum_{I,J=m+1}^{M} |I\>\<J| \nn \\
%& +\frac{\sin\a_r\cos\a_r}{\sqrt{m(N-m)}}\sum_{i=1}^{m}\sum_{I=m+1}^{N}\Big( |i\>\<I| +|I\>\<i| \Big)
%\end{align}

Then the success probability for finding target states is
\begin{align}
P(r)= \sin^2\a_r.
\end{align}
 
The states after $r$ times of iteration gives a density matrix, and the authors of \cite{shi} calculated the relative entropy of coherence and $l_1$-norm coherence with it. They showed that during the success probability $P(r)$ increases from $0$ to $1$, the amounts of coherence decrease monotonically. These phenomena support the conjecture that coherence is a key resource for Grover search process. 

\subsection*{$r_C$ and $C_c^{(N)}$ as resources for Grover search}  

It is quite straightforward to see the change of coherence number of Eq. \eqref{pr} along $r$. Since the state is pure, the coherence number is just the coherence rank.  $r_C(|\p_r\>)$ remains constant until $r$ exactly satifies $\cos\a_r =0$, i.e.,
\begin{align}
\label{drop}
&r\neq \Big( \frac{\pi}{2A }-\frac{1}{2}\Big): \quad 0 \le P(r) < 1, \quad r_C(|\p_r\>)=N \nn \\
&r= \Big( \frac{\pi}{2A}-\frac{1}{2}\Big): \quad P(r)=1, \quad r_C(|\p_r\>)=m 
\end{align}
The coherence number of $|\p_r\>$ suddenly drops down to $m$ (the number of target states) from $N$ when $r$ reaches $\frac{\pi}{2A }-\frac{1}{2} $. So we can say that \emph{the leaping off of coherence number is an alarm bell to notice that $P(r)$ has reached its maximal value exactly}. But since it usually does not happen that $\frac{\pi}{2A}-\frac{1}{2}$ becomes an integer, we can say for most cases that $r_C(|\p_r\>)$ remains $N$ throughout the searching process.

One thing to pay attention is that the final state after finishing the searching task, even when $\frac{\pi}{2A }-\frac{1}{2} $ is an integer, is still coherent except when $m=1$. We can see the same pattern in Figure 2 of \cite{shi}, which shows that the relative entropy of coherence $C_r(|\p_r\>)$ is  still non-zero at $P=1$. The same is true with the $l_1$-norm monotone and the geometric coherence \cite{rast}. If there are coherence monotones which the iteration of $G$ depletes completely at $P=1$, we can say that they are the optimal measures of coherence consumption during the searching process.

As such a monotone, we introduce $C_c^{(N)}$, the last member of the \emph{generalized coherence concurrence} \cite{Chin2}. It is an analogous coherence monotone family to the generalized entanglement concurrence and  consists of coherence $k$-concurrences with $2\le k \le N$ ($N$ is the dimension of state here). The family is coherence number specific, just as the generalized entanglement concurrence is Schmidt number specific. So Eq. \eqref{drop} motivates us to consider $C_c^{N}(|\p_r\>)$ as an optimal measure, for $C_c^{(N)}(\r) \neq 0$ if and only if $r_C(\r) =N$.

  While the general definition for the whole members of the monotone family is given in \cite{Chin2}, here we just need the definition for $C_c^{(N)}$:
 \begin{definition}
 	For a pure state $|\p\>=\sum_{i=1}^{N}\p_i|i\>$ ($\{|i\>\}_{i=1}^{N}$ is the computational  basis set), 
 	 \begin{align}
 	C_c^{(N)}(|\p\>)\equiv N  \Big|\p_1^2\p_2^2\cdots \p_N^2 \Big|^{\frac{1}{N}}
 	\end{align}
  and $C_c^{(N)}(\r)$ for a mixed state $\r$ is obtained by convex roof extention.
 \end{definition}
 $C_c^{N}$ is a normalized monotone, i.e.,  $C_c^{(N)}(\r)=1$ when $\r$ is maximally coherent. It is clear that $C_c^{(N)}(\r) \neq 0$ if and only if $r_C(\r) =N$ from the form of the definition.

For our case the state is pure and $|\p\> =|\p_r\>$, so $C_c^{(N)}(|\p_r\>)$ is given by
\begin{align}\
\label{CcN}
C_c^{(N)}(|\p_r\>)=N \Big(\frac{\sin^2\a_r}{m}\Big)^\frac{m}{N} \Big(\frac{\cos^2\a_r}{N-m}\Big)^{\frac{N-m}{N}}.
\end{align}
We first check the values of $C_c^{(N)}(|\p_r\>)$ at $r=0$ and $r=\frac{\pi}{2A}-\frac{1}{2}$,
\begin{align}
\label{boundary}
&r=0: \quad \tan^2\a_0 =\frac{m}{N-m}, \quad C_c^{(N)}(\p_0\>)=1 \nn \\
&r=\Big(\frac{\pi}{2A}-\frac{1}{2}\Big): \quad \cos\a_r=0, \quad  C_c^{(N)}(|\p_r\>)=0 
\end{align}
$C_C^{(N)}(|\p_r\>)$ completely goes away when $P(r)=1$ as expected. We obtain the behavior of $C_c^{(N)}$ in the midway between  $r=0$ and $r=\frac{\pi}{2A}-\frac{1}{2}$ by differentiating $C_c^{(N)}$ with $r$,

\begin{align}
\frac{d C_c^{(N)}}{d r} = &2A\Big( \frac{\sin^{2m}\a_r \cos^{2(N-m)}\a_r}{m^m(N-m)^{(N-m)}}\Big)^{\frac{1}{N}}  \nn \\
& \times \frac{\Big(m\cos^2\a_r -(N-m)\sin^2\a_r\Big)}{\sin\a_r\cos\a_r} \le 0.
\end{align}
The last inequality comes from $\tan^2\a_r \ge \frac{m}{N-m}$. As a result, \emph{$C_c^{(N)}(|\p_r\>)$ is a monotonically decreasing function of $r$ from $1$ to $0$ and completely consumed to perform the Grover search process}. The case for $N=2^{10}$ and $m=5$ is ploted in Fig. 1.

\begin{figure}
	\label{graph}
	\centering
	\includegraphics[width=8cm]{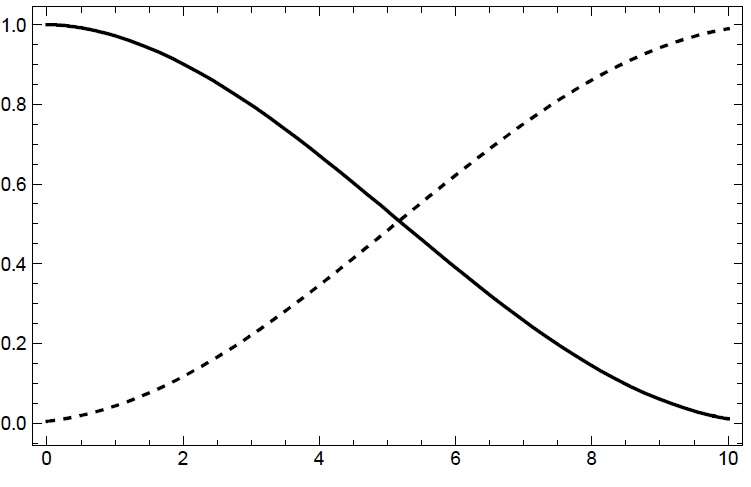}
	\caption{The change of $C_c^{(2^{10})}$ (solid line) and P(r) (dashed line) with $N=2^{10}$ and $m=5$ from $r=0$ to $r=10$.}
\end{figure}

We can also  calculate the cost performance $w=-dP/dC$ for $C=C_c^{(N)}(|\p_r\>)$. Actually,  Eq. \eqref{CcN} is re-expressed with $P$ as
\begin{align}
C_c^{(N)}(P)=N \Big(\frac{P}{m}\Big)^\frac{m}{N} \Big(\frac{1-P}{N-m}\Big)^{1-\frac{m}{N}},
\end{align} 
so we have
\begin{align}
w = \frac{m^\frac{m}{N}(N-m)^{1-\frac{m}{N}} (1-P)^{\frac{m}{N}} P^{1-\frac{m}{N}} }{N\Big( P- \frac{m}{N}\Big) } \ge 0
\end{align}
by $P =\sin^2\a_r \ge m/N$. The cost performance $w$ is very high when $r$ is small and  goes to 0 at $P=1$.
%This gives the upper and lower bound of $C_c^{(N)}(|\p_r\>) + P(r)$ as
%\begin{align}
%1 \le C_c^{(N)}(P) +P \le 2.
%\end{align}
%The upper bound is trivial, and the lower bound comes from the fact that $C_c^{(N)}$ is a decreasing concave function for $\frac{m}{N} \le P\le 1$ which gives $C_c^{(N)}(P) \ge 1-P$.  The case for $N=2^{10}$ and $m=5$ is ploted in Fig. 2.
When $N \gg 1$ and $m \ll N$, the above equation is simplified to a function of $P$ and $\frac{m}{N}$ as
\begin{align}
w \to \frac{ \Big(\frac{m}{N}\Big)^\frac{m}{N} (1-P)^\frac{m}{N} P^{1-\frac{m}{N}} }{P- \frac{m}{N}}.
\end{align}

Before closing this section, we roughly sketch the behavior of coherence $k$-concurrences $C_c^{(k)}(|\p_r\>)$ with $ m+1 \le k \le N-1$. All members $C_c^{(k)}(\r)$ in the generalized coherence concurrences are normalized and nonzero if and only if $r_C(\r) \ge k$ \cite{Chin2}. So their boundary conditions along $r$ including $k=N$ are expressed as
\begin{align}
&(\textrm{For all $k$'s such that } m+1 \le k \le N)\nn \\
&r=0: \quad \tan^2\a_0 =\frac{m}{N-m}, \quad C_c^{(k)}(\p_0\>)=1 \nn \\
&r=\Big(\frac{\pi}{2A}-\frac{1}{2}\Big): \quad \cos\a_r=0, \quad  C_c^{(k)}(|\p_r\>)=0 
\end{align}

  So we can say that coherence $k$-concurrences with $m+1 \le k \le N$ are completely consumed during the Grover search process.

\section{Conclusions}\label{fin}
\indent

In summary, we introduced the coherence number $r_C(\r)$ for mixed states and obtained a necessary and sufficient condition for a coherent mixed state to be converted to a bipartite entangled state of nonzero $k$-concurrence. We also showed that the coherence number is a simple and clear measure for the success probability  of the Grover search process and that the  continuous monotone $C_c^{(N)}$ is thoroughly exploited to finish the task.

Considering the relation between the Schmidt number and the $k$-concurrence in entanglement, it is natural to expect there exists a family of coherence concurrences which senses the coherence number directly, which is introduced in \cite{Chin2} (the coherence $k$-concurrence $C_c^{(k)}(\r)$ of a $d$-dimensional state $\r$ with $2\le k \le d$). In the paper, the application of $r_C(\r)$ and the concurrence family to the multi-slit interference experiments is also presented. But while the coherence number determines the number of distinguishable slits and $C_c^{(2)}$ can be understood as a kind of visibility, the quantitative meaning of $C_c^{(k)}$ with $k\neq 2$ in the multi-slit problem is not clear yet. Considering the role of the general coherence concurrence in Grover algorithm, the monotonicity of $C_c^{(k)}$ with $m+1\le k \le N-1$ during the searching process is also to be studied.

It will also be an intriguing problem to find a more systematic and geometric way of understanding the relations among the Schmidt number, the coherence number, and the generalized concurrences of entanglement and coherence.

\section*{Acknowledgements}
\indent

The author is grateful to Prof. Jung-Hoon Chun for his support during the research, and the anonymous
referee for advising on the improvement of the paper. This was supported by Basic Science Research Program through the National Research Foundation of Korea funded by the Ministry of Education(NRF-2016R1D1A1B04933413).

\appendix

\section{Axioms that coherence monotones should fulfill}\label{off} Coherence monotones should satisfy the following conditions \cite{baum}:\\
\\
(C1) Nonnegativity: $C(\r) \ge 0$

(stronger condition: $C(\r) =0$ if and only if $\r$ is incoherent)\\
(C2) Monotonicity: $C(\r)$ does not increase under the incoherent operations, i.e.,
$C(\La[\r]) \le C(\r)$ for any incoherent operation  $\La$, where $\La: \mathcal{B(H)} \mapsto \mathcal{B(H)}$ permits a set of Kraus operators $\{K_n \}$  such that $\sum_n K_n^\dagger K_n=\mathbb{I}$ and $K_n\d K_n^\dagger$ $\in \mathcal{I}$ for any $\d \in \mathcal{I}$ (the set of incoherent states, expressed as $\r=\sum_{i=1}^{d}p_i|i\>\<i|$).  \\
(C3) Strong monotonicity: $C$ does not increase under selective incoherent operations, i.e., $\sum_np_nC(\r_n) \le C(\r)$ with $p_n=tr[K_n\r K_n^\dagger]$, $\r_n=K_n\r K_n^\dagger/p_n$ for incoherent Kraus operators $K_n$.\\
(C4) Convexity: $\sum_ip_iC(\r_i) \ge C\Big(\sum_ip_i\r_i\Big)$.
\\

A quantity should fulfill at least (C1) and (C2) to be a coherent monotone, and if (C3) and (C4) are satified then (C2) is automatically satisfied. 

The incoherent Kraus operators are expressed more explicitly from the condition $K_n|j\> \sim |k\>$ ($|j\>$ and $|k\>$ are both in the computational basis set $\{|i\> \}^d_{i=1}$) for each $n$ as
\begin{align}
\label{kraus}
K_n = \sum_{i=1}^{d} c_n^i |s_i^n\>\< i|, 
\end{align}
where $s_i^n$ is a function that sends $i$ to a number from $1$ to $d$ so that $|s_i^n\>$ is in $\{|i\>\}_{i=1}^d$ and $\sum_{j=1}^{d} \<i|s_j^n\> =1$ \cite{winter}. 
Then the normalization condition for $K_n$ 
\begin{align}
\sum_nK_n^\dagger K_n = \sum_{i,j} \Big( \sum_n c_n^{i*}c_n^{j}\<s_i^n|s_j^n\> \Big)|i\>\<j| =\sum_i|i\>\<i| 
\end{align}
gives
\begin{align} 
\sum_n c_n^{i*}c_n^{j}\<s_i^n|s_j^n\> = \d_{ij} , \qquad \sum_n|c_n^i|^2 =1  \textrm{ for each } i .
\end{align}

%\section{Identities for the coherence concurrence}\label{id}
% For a $(d\times d)$-dimensional pure state $|\p\> = \sum_{ij}\p_{ij}|ij\>_{SA}$, we have
%\begin{align}
%C_c(|\p\>) &=\sum_{(i,j)\neq (k,l)}|\p_{ij}\p_{kl}^*| \nn \\
%&= \Big(\sum_{i,j}r_{ij}\Big)^2 -\sum_{i,j}r_{ij}^2  = \Big(\sum_{i,j}r_{ij}\Big)^2 -1 \qquad (r_{ij}\equiv |\p_{ij}|) \nn \\
%&=\sum_{i\neq k}\sum_{ j\neq l}r_{ij}r_{kl} +\sum_{i}\sum_{j\neq l}r_{ij}r_{il} +\sum_{i\neq k}\sum_{j}r_{ij}r_{kj} \nn \\
%&=2\Big(\sum_{i< k}\sum_{ j< l}r_{ij}r_{kl} +\sum_{i< k}\sum_{ j<l}r_{il}r_{kj} +\sum_{i}\sum_{j< l}r_{ij}r_{il} +\sum_{i< k}\sum_{j}r_{ij}r_{kj} \Big)\nn \\
%\end{align}
\section{The proof of \eqref{Cc3}}\label{cc3}
Since the inequality
\begin{align}
C_c(\r^S)=C_c(\r^s\otimes |1\>\<1|^A) \ge C_c(\La^{SA} [ \r^s\otimes |1\> \<1|^A])
\end{align}
is clear,  what we need to prove is
\begin{align}
& \Big(\frac{3d^2}{4(d-1)(d-2)}\Big)^\frac{1}{3} C_c(\La^{SA} [ \r^s\otimes |1\> \<1|^A])\nn \\
&\qquad  \ge C_3( C_c(\La^{SA} [ \r^s\otimes |1\> \<1|^A]) ).
\end{align}
Expending \eqref{c3} as
\begin{align}
C_3(|\p\>) = & \Bigg[\frac{3!d^2}{(d-1)(d-2)} \nn \\
& \times\sum_{\substack{i < k < m \\j<l<n}} \Big|\p_{ij}\p_{kl}\p_{mn} +\p_{il}\p_{kn}\p_{mj}  +\p_{in}\p_{kj}\p_{ml} \nn \\ 
& \qquad\quad -\p_{in}\p_{kl}\p_{mj} - \p_{il}\p_{kj}\p_{mn}  -\p_{ij}\p_{kn}\p_{ml}  \Big|^2 \Bigg]^{\frac{1}{3}},
\end{align}
we have
\begin{align}
& \frac{1}{d^3}\binom{d}{3}\Big(C_3(|\p\>)\Big)^3 \nn \\
& \le  \sum_{i<k<m}\sum_{j<l<n}  \Big(r_{ij}r_{kl}r_{mn} +r_{il}r_{kn}r_{mj}  +r_{in}r_{kj}r_{ml} \nn \\ 
&\qquad\qquad\qquad\quad +r_{in}r_{kl}r_{mj}+ r_{il}r_{kj}r_{mn}  +r_{ij}r_{kn}r_{ml}  \Big)^2,
\end{align}
where  $ r_{ij} \equiv |\p_{ij}|$. Then using 
\begin{align}
&\sum_{\substack{ i<k<m \\j<l<n}} r_{ij}^2r_{kl}^2r_{mn}^2 \nn \\
&\quad =\sum_{\substack{(i<k, j<l) \\(m<p, n<q)\\(r<v,s<w)}} r_{ij}r_{kl} r_{mn}r_{pq} r_{rs}r_{vw} \d_{km}\d_{ln}\d_{nv}\d_{qw}\d_{ir}\d_{js},
\end{align}
\begin{align}
& \sum_{\substack{i<k<m\\j<l<n}} r_{ij}r_{kl}r_{mn}r_{il}r_{kn}r_{mj}  \nn \\
& = \sum_{\substack{(i<k, j<l)\\(m<p, n<q)\\(r<v,s<w)}}  r_{ij}r_{kl} r_{mn}r_{pq} r_{rs}r_{vw}
\d_{ip}\d_{lq}\d_{kr}\d_{ns}\d_{mv}\d_{jw} 
\end{align}
and so on, we have 
\begin{align}
&\frac{2^3}{d^3}\binom{d}{3}\Big(C_3(|\p\>)\Big)^3 \nn \\
& \le 2^3\Big(\sum_{(i<k,j<l)}r_{ij}r_{kl} +\sum_{(i<k,j<l)}r_{il}r_{kj} \nn \\& \qquad\quad\qquad +\sum_{(i,j<l)}r_{ij}r_{il} +\sum_{(i<k,j)}r_{ij}r_{kj} \Big)^3 \nn \\
&=C_{c}(|\p\>)^3.
\end{align}
By convex roof extension, we have \eqref{Cc3}.

\bibliography{conum2}

 \end{document}